\documentclass[pra,aps,a4paper,showpacs,superscriptaddress,floatfix,twocolumn]{revtex4-1}
\usepackage{amsmath}
\usepackage{amsthm}
\usepackage{amsfonts}
\usepackage{amssymb}
\usepackage{graphicx}
\usepackage{dcolumn}
\usepackage{bm}

\newtheorem{Lemma}{Lemma}
\newtheorem{Proposition}{Proposition}

\DeclareMathOperator{\Span}{span}
\DeclareMathOperator{\sgn}{sgn}

\begin{document}

\title{Optimal Slater-determinant approximation of fermionic wave functions}

\author{J.~M.~Zhang}
\email{wdlang06@163.com}
\affiliation{Fujian Provincial Key Laboratory of Quantum Manipulation and New Energy Materials,
College of Physics and Energy, Fujian Normal University, Fuzhou 350007, China}
\affiliation{Fujian Provincial Collaborative Innovation Center for Optoelectronic Semiconductors and Efficient Devices, Xiamen, 361005, China}

\author{Norbert J. Mauser}
\email{mauser@courant.nyu.edu}
\affiliation{Wolfgang Pauli Institute c/o Fak. f. Mathematik,
Univ.~Wien, Oskar-Morgenstern Platz 1, 1090 Vienna, Austria}

\begin{abstract}
We study the optimal Slater-determinant approximation of an $N$-fermion wave function analytically. That is, we seek the Slater-determinant (constructed out of $N$ orthonormal single-particle orbitals) wave function having largest overlap with a given $N$-fermion wave function. Some simple lemmas have been established and their usefulness is demonstrated on some structured states, such as the Greenberger-Horne-Zeilinger state. In the simplest nontrivial case of three fermions in six orbitals, which the celebrated Borland-Dennis discovery is about, the optimal Slater approximation wave function is proven to be built out of the natural orbitals in an interesting way. We also show that the Hadamard inequality is useful for finding the optimal Slater approximation of some special target wave functions.
\end{abstract}

\maketitle

\section{Introduction}
Fermionic wave functions have deep structures. A commonplace knowledge as a consequence of the antisymmetry condition \cite{dirac,heisenberg} is the Pauli exclusion principle \cite{pauli} (though historically, the two were developed in the reversed order). However, less appreciated is that the antisymmetry condition can have much deeper consequences. In this regard, we have the celebrated discovery by Borland and Dennis \cite{borland1, borland2, ruskai}: For a wave function with three fermions in six orbitals, the occupation numbers $\lambda_i$ in the six natural orbitals \cite{lowdin}, ordered in the descending order, satisfy
\begin{subequations}\label{borlandanddennis}
\begin{eqnarray}
\lambda_1 + \lambda_6  = \lambda_2 + \lambda_5 =\lambda_3 + \lambda_4 =1, \\
\lambda_5 + \lambda_6 \geq \lambda_4, \quad \quad \quad\quad  \quad
\end{eqnarray}
\end{subequations}
besides the common knowledge of $0\leq \lambda_i \leq 1 $.
It is remarkable that initially this discovery was made through numerical experiments---apparently, it is hard to suspect such relations analytically.
The fact that the $m$-particle reduced density matrix of an $N$-fermion wave function must have nontrivial structures
had actually been perceived earlier by Coleman \cite{coleman}. This led him to propose the ``$N$-representability'' problem \cite{coleman, yukalov}, which turns out to be a very difficult one.
Only recently has the problem in the simplest case of $m=1$ been solved \cite{turkey},
generalizing the equalities and inequalities in (\ref{borlandanddennis}) systematically. This has led to a burst of studies of the relevance and implications of the so-called generalized Pauli constraints like (\ref{borlandanddennis}) in atoms, molecules, and model systems \cite{schilling13,schilling15, tennie15, spring13, spring15, spring14, chakraborty14, chakraborty15,chakraborty16, chakraborty152}. As for the next case of $m=2$ (which is possibly of more interest from the point of view of calculating the ground state energy of a multi-electron system), a systematic procedure for generating the $N$-representability conditions on the two-particle reduced density matrix has been derived by Mazziotti \cite{mazziotti12}. In this case, also, approximate $N$-representability conditions have been applied to a variety of systems \cite{mazziotti13}.

In this paper, we study the structure of an $N$-fermion wave function by looking for its optimal Slater approximation. The idea is to approximate an antisymmetric wave function with the simplest kind of wave function satisfying the antisymmetry condition \cite{slater}, so as to gain an idea of the complexity of the target wave function. Here by optimal \cite{hartreefock}, we mean the overlap between the target wave function and the Slater wave function is maximized.
In view of the fact that a Slater state is determined by $M=N$ orthonormal single-particle orbitals (or more appropriately, an $M=N$ dimensional subspace of the single-particle Hilbert space), the problem can be generalized to $M>N$. That is, we seek $M$ orthonormal single-particle orbitals (or an $M$-dimensional subspace) so that the projection of the wave function onto the space spanned by the  $C_M^N$ Slater determinants is maximized.
Apparently, this problem is about the geometry of a multifermion Hilbert space. It naturally provides a geometric measure \cite{wei,alex,alex2} of the entanglement between the fermions. Compared to those measures based on Schmidt decomposition or dividing a multipartite system into subsystems, it has the advantage of identical particles being treated identically, namely, on an equal footing \cite{ugo}.

The problem was proposed and studied previously in Ref.~\cite{zjm},
with the motivation of gauging the reliability of the multi-configuration time-dependent Hartree-Fock (MCTDHF) algorithm
\cite{zanghellini,caillat, kato, nest, alon, BardosGGM, BardosCMT}.
There, the approach was primarily numerical, and an efficient iterative algorithm was proposed.
On the analytic side, the  $N=2$ case turned out to be simple enough to allow for a complete solution.
The answer is simply that one should take the $M$ (if $M$ is even; otherwise, take $M-1$) most occupied natural orbitals.

We also note that the structure of a fermionic wave function has recently been extensively studied by Chen \textit{et al}.
\cite{lin2015,lin3in6}. Some of their results are of direct relevance to our problem.
Their approach has the merit of being more systematic and more general, but sometimes it is too sophisticated.

This paper is organized as follows. After stating the problem in Sec.~\ref{formulation},
we will establish in Sec.~\ref{lemmas} some simple lemmas.
These lemmas, although simple, enable us to find the optimal Slater approximation of some simply structured states effortlessly.
For example, for the state $f= (\sqrt{2}|123 \rangle +  |456 \rangle)/\sqrt{3}$,
we can immediately tell that its optimal Slater approximation is $|123\rangle $.
Then in Sec.~\ref{threeinsix}, we will turn from the general case to the particular case of three fermions in six orbitals,
i.e., we will focus on states $f\in \wedge^3 \mathcal{H}_6$.
This is the simplest nontrivial case both in Borland-Dennis's study and in ours.
We will prove the canonical form of the wave function, which was first proven by Chen \textit{et al}. \cite{lin3in6},
using a more elementary and more straightforward method.
With this canonical form, we find that the optimal Slater approximation has a simple structure in terms of the natural orbitals.
In Sec.~\ref{hadamardsection}, we show that the Hamadard inequality is useful for determining the optimal Slater approximation
for many structured states of interests. We summarize in Sec.~\ref{conclusion} and discuss some open questions.

\section{Formulation of the problem}\label{formulation}
The problem has been formulated in Ref.~\cite{zjm} before.
Here we reformulate it in terms of exterior algebra, or Grassmann algebra, which is the most suitable language for fermions \cite{textbook1}.

For $N$ particles in a $d$-dimensional Hilbert space $ \mathcal{H}_d$,
the total Hilbert space is the tensor space $T^N (\mathcal{H}_d) \equiv \otimes^N \mathcal{H}_d $. The inner product on $T^N $ is defined by
\begin{eqnarray}
\langle \phi_1 \otimes \ldots \otimes \phi_N | \psi_1 \otimes \ldots \otimes \psi_N \rangle_T = \prod_{i=1}^N (\phi_i, \psi_i ),
\end{eqnarray}
and extended by linearity. Here $\phi_i$, $\psi_j \in \mathcal{H}_d$ and $(\cdot, \cdot )$ is the inner product on $\mathcal{H}_d$.

If the particles are identical fermions, the wave function must be an alternating, or antisymmetric, tensor. The effective Hilbert space is then the alternating subspace $A^N(\mathcal{H}_d)\equiv \wedge ^N\mathcal{H}_d$, which is spanned by tensors in the form of
\begin{eqnarray}
\phi_1 \wedge \ldots \wedge \phi_N  &\equiv & \mathcal{A}_N (\phi_1 \otimes \ldots \otimes \phi_N)  \nonumber \\
&=& \frac{1}{N!}\sum_{P\in S_N } \sgn(P) \phi_{P_1}\otimes \ldots \otimes \phi_{P_N}.
\end{eqnarray}
Here $\sgn(P)$ means the sign of the permutation $P$.
If $\{ e_1, e_2, \ldots, e_d \}$ is a linearly independent basis of $\mathcal{H}_d$, then $\{ e_{j_1} \wedge e_{j_2}\wedge \ldots \wedge e_{j_N} \}$ with $1\leq j_1 < j_2 \ldots < j_N \leq d $ is a linearly independent basis of $\wedge^N \mathcal{H}_d$.

In mathematical terms, an $N$-fermion wave function $f$ in $\wedge^N \mathcal{H}_d$ is called an $N$-vector. If there exist $N$ linearly independent vectors (or orbitals, in the physical term) $\{ \varphi_i | 1\leq i \leq N \}$ such that
\begin{eqnarray}
f = \varphi_1 \wedge \varphi_2 \wedge \ldots \wedge \varphi_N,
\end{eqnarray}
we call $f$ a decomposable $N$-vector. Physically, it is a Slater-determinant wave function constructed out of the orbitals $\{ \varphi_i | 1\leq i \leq N \}$.
A basic point is that a Slater determinant is determined by the subspace spanned by the orbitals, i.e., $\Span \{ \varphi_1, \ldots, \varphi_N \}$, up to a nonzero scalar factor. Actually, suppose $\{ \psi_i = \sum_{j=1}^N a_{ij} \varphi_j |1\leq i\leq N  \} $ is another linearly independent basis of $\Span \{ \varphi_1, \ldots, \varphi_N \}$, which means the transform matrix $a$ is nonsingular, then
\begin{eqnarray}
\psi_1 \wedge \ldots \wedge \psi_N = \det(a) \varphi_1 \wedge \ldots \wedge \varphi_N.
\end{eqnarray}

The inner product on $\wedge^N \mathcal{H}_d$ is defined as
\begin{eqnarray}
& & \langle \phi_1 \wedge \ldots \wedge \phi_N | \psi_1 \wedge \ldots \wedge \psi_N\rangle\nonumber \\
&\equiv & N! \langle \phi_1 \wedge \ldots \wedge \phi_N | \psi_1 \wedge \ldots \wedge \psi_N\rangle_T = \det [(\phi_i, \psi_j )],\quad
\end{eqnarray}
and extended by linearity. We easily see that two Slater determinants are orthogonal if and only if there is a nonzero vector in $V_{\phi}=\Span \{\phi_1, \ldots, \phi_N \}$ orthogonal to the space $V_{\psi}= \Span \{\psi_1, \ldots, \psi_N \} $, or equivalently, there is a nonzero vector in $V_\psi$ orthogonal to $V_\phi$.

For an arbitrary $N$-vector $f$ in $\wedge^N \mathcal{H}_d$, a subspace $V \subset \mathcal{H}_d$ is said to envelop $f$ if $f\in \wedge^N V$, with $\wedge^N V $ regarded naturally as a subspace of $\wedge^N \mathcal{H}_d$. In other words, $f$ is a linear combination of exterior products $\phi_1^{(i)} \wedge \ldots \wedge \phi_N^{(i)}$, with $\phi_k^{(i)} \in V$. It can be shown that there exists a minimal subspace $\mathcal{E}(f)$ which is a subspace of every $V$ enveloping $f$. Actually, $\mathcal{E}(f) $ is spanned by the natural orbitals of $f$ with nonzero occupation numbers. The dimensional of $\mathcal{E}(f)  $ is called the rank of $f$. In this term, an $N$-vector $f$ is decomposable if and only if its rank is $N$.

Now comes the problem. An $N$-fermion wave function $\Psi $ in $\wedge^N \mathcal{H}_d$, say the ground state of a system of interacting fermions,
is generally not a Slater determinant, or not a decomposable $N$-vector. However, a natural question out of the spirit of approximation is, is it possible to approximate the wave function with a Slater determinant to a good accuracy? The inner product on $\wedge^N \mathcal{H}_d$ provides a natural measure of distance
or proximity, without reference to any physical quantity. The quantity of interest is then
\begin{eqnarray}
I_\text{max}(\Psi;N) \equiv  \max_{V_\phi \subset \mathcal{H}_d } \left |\langle \phi_1 \wedge \phi_2 \wedge \ldots \wedge \phi_N | \Psi \rangle  \right|^2.
\end{eqnarray}
Here the maximum is taken over all $N$-dimensional subspaces $V_\phi $ of $\mathcal{H}_d$, with $\phi_1, \phi_2, \ldots, \phi_N $ being its orthonormal basis.
This is the simplest, single-configuration approximation, with the rank of the approximation $M$ equal to the number of fermions.

In the more general case, the rank of the approximation $M$ can be larger than $N$. The quantity of interest is then
\begin{eqnarray}
I_\text{max}(\Psi;M) &\equiv & \max_{V_\phi \subset \mathcal{H}_d } \sum_{J}\left |\langle \phi_{j_1} \wedge \phi_{j_2} \wedge \ldots \wedge \phi_{j_N} | \Psi \rangle  \right|^2. \quad
\end{eqnarray}
Here $J= (1\leq j_1 < j_2<\ldots < j_N\leq M )$ runs through all the different $N$-tuples. The maximum is taken over all $M$-dimensional subspaces $V_\phi $ of $\mathcal{H}_d$, with $\phi_1, \phi_2, \ldots, \phi_M $ being its orthonormal basis.

In this paper, we shall focus primarily  on the single-configuration case with $M=N$, but occasionally we also touch on the multi-configuration case.

\section{Useful Lemmas}\label{lemmas}
%

First we formulate some lemmas (some of which were known previously) and give self-consistent proofs:

\begin{Lemma}\label{hole}
If $d= N+1$, then for an arbitrary normalized wave function $\Psi \in \wedge^N \mathcal{H}_d$,  $I_\text{max} (\Psi ; N ) = 1$.
That is, $\Psi$ can be written in the form of a Slater wave function.
\end{Lemma}
\begin{proof}
First we show that a wave function in the form of
\begin{eqnarray}\label{W1}
W = \phi_1 \wedge \phi_2 \wedge \ldots \wedge \phi_N + \psi_1 \wedge \psi_2 \wedge \ldots \wedge \psi_N,
\end{eqnarray}
where $\{ \phi_i | 1\leq i \leq N \}$ and $\{ \psi_i | 1\leq i \leq N \}$ are two sets of linearly independent vectors in $\mathcal{H}_d$, can be rewritten in the compact form of a single Slater determinant,
\begin{eqnarray}\label{W2}
W = \varphi_1 \wedge \varphi_2 \wedge \ldots \wedge \varphi_N.
\end{eqnarray}
Consider the union of the two $N$-dimensional spaces $V_\phi \equiv \Span\{\phi_1, \ldots, \phi_N\}$ and $V_\psi \equiv \Span\{\psi_1, \ldots, \psi_N\}$. If $\dim ( V_\phi \cup V_\psi) = N$, then $V_\phi = V_\psi$, and the two terms in (\ref{W1}) are the same up to a scalar factor and thus $W$ can be rewritten in the form of (\ref{W2}). If $\dim ( V_\phi \cup V_\psi) = N+1$, then $\dim (V_\phi \cap V_\psi )= N-1$. Let $V_\phi \cap V_\psi  = \Span \{ \varphi_1, \ldots, \varphi_{N-1} \}$. There exist $\tilde{\phi}_N$ and $\tilde{\psi}_N$ such that
\begin{eqnarray}
\phi_1 \wedge \phi_2 \wedge \ldots \wedge \phi_N = \varphi_1 \wedge \ldots \wedge \varphi_{N-1}\wedge \tilde{\phi}_N , \\
\psi_1 \wedge \psi_2 \wedge \ldots \wedge \psi_N = \varphi_1  \wedge\ldots \wedge \varphi_{N-1} \wedge \tilde{\psi}_N .
\end{eqnarray}
Therefore,
\begin{eqnarray}
W = \varphi_1 \wedge \ldots \wedge \varphi_{N-1}\wedge ( \tilde{\phi}_N  + \tilde{\psi}_N ),
\end{eqnarray}
which is in the form of (\ref{W2}).

Now an arbitrary state $\Psi \in \wedge^N \mathcal{H}_d$ can be expanded as
\begin{eqnarray}
\Psi = \sum_{i=1}^d  c_i e_1 \wedge \ldots \wedge \hat{ e}_i \wedge \ldots \wedge e_d ,
\end{eqnarray}
where $\hat{e}_i $ means $e_i$ is missing in the wedge product. Now by carrying out the process of merging two Slater determinants into one, we can reduce $\Psi$ to the form of a single Slater determinant in at most $N$ steps. Therefore, an arbitrary state $\Psi \in \wedge^N \mathcal{H}_d $ can always be reduced to a Slater-determinant form and thus $I_\text{max} (\Psi ; N ) = 1$.
\end{proof}

This lemma is already stated in Ando's paper \cite{ando}. In Ref.~\cite{zjm}, it was referred to as the $N$-in-$(N+1)$ theorem and several proofs were given. Here, yet another proof is given just for the sake of completeness. As emphasized in Ref.~\cite{zjm}, a consequence of Lemma \ref{hole} is for an arbitrary wave function $\Psi \in \wedge^N \mathcal{H}_d$, where $d$ is not constrained to $N+1$ but can be arbitrary, $I_\text{max}(\Psi ; N )= I_\text{max} (\Psi; N+1)$, since a $(N+1)$-rank approximation is actually an $N$-rank approximation by the lemma.

\begin{Lemma}\label{least}
If $d>N$, and $M = d -1$, that is, we consider dropping one orbital, then the orbital which should be dropped is the least occupied natural orbital, and $I_\text{max}(\Psi;d-1)= 1 - \lambda_d $, where $\lambda_d$ is the occupation of the least occupied natural orbital.
\end{Lemma}
\begin{proof}
Let $\{ e_1, \ldots, e_d \}$ be an orthonormal basis of $\mathcal{H}_d$, in which $e_d$ will be dropped. The wave function can be expanded as
\begin{eqnarray}
\Psi = \sum_J C_J e_{j_1} \wedge e_{j_2} \wedge \ldots \wedge e_{j_N}.
\end{eqnarray}
Here $J$ denotes an ordered $N$-tuple $(1\leq j_1<j_2 \ldots < j_N \leq d )$. We have
\begin{eqnarray}
I (\Psi;d-1) &=& 1 - \sum_{J|_{d \in J}} |C_J|^2 = 1- \langle e_d | \rho | e_d \rangle ,
\end{eqnarray}
where $\rho$ is the one-particle reduced density matrix of $\Psi$. Apparently, to maximize $I$, we have to take $e_d$ as the eigenvector of $\rho$ corresponding to the smallest eigenvalue, or the least occupied natural orbital.
\end{proof}

This lemma is also known previously \cite{holland}.

\begin{Lemma}\label{factor2}
Let $\Psi$ be a normalized $N$-fermion wave function in which orbital $f$ is occupied with probability unity. That is, $\Psi = f \wedge \Psi' $, with $\Psi'$ being a $(N-1)$-fermion wave function and $f \perp \mathcal{E}(\Psi')$. Then
\begin{eqnarray}\label{factor3}
I_{\text{max}} (\Psi; N) = I_{\text{max}} (\Psi'; N-1).
\end{eqnarray}
Moreover, the optimal Slater approximation $S$ of $\Psi$ is $S= f \wedge S'$, with $S'$ being the optimal Slater approximation of $\Psi'$.
\end{Lemma}
\begin{proof}
For the optimal Slater approximation $S$, we can find an orthonormal basis $\{\phi_1, \phi_2, \ldots, \phi_N\}$ for its supporting space $V$ such that
\begin{eqnarray}\label{ortho1}
\langle f | \phi_i \rangle=0, \quad 2\leq i \leq N,
\end{eqnarray}
as follows. Decompose $f$ as $f= f_\parallel + f_\perp$ with $f_\parallel \in V$ and $f_\perp\in V^\perp $. If $f_\parallel \neq 0$, take $\phi_1 \propto f_\parallel$, while if $f_\parallel =0$, take $\phi_1$ arbitrarily. Basis vectors $\{\phi_2, \ldots, \phi_N\}$ can then be obtained by extending $\phi_1$ into a full orthonormal basis of $V$. It is readily seen that such chosen basis vectors satisfy condition (\ref{ortho1}).

We have then $S = (a f + \tilde{\phi}_1) \wedge \Phi$,
where $\phi_1 = a f + \tilde{\phi}_1$, $a = \langle f | \phi_1 \rangle $, and $\Phi = \phi_2 \wedge \ldots \wedge \phi_N$. Now
\begin{eqnarray}\label{ineqn1}
I_{\text{max}} (\Psi; N) = |\langle S | \Psi \rangle |^2 &=& |a \langle \Phi | \Psi'\rangle |^2
 \leq  |\langle  \Phi |\Psi'\rangle |^2 \nonumber \\
&\leq & I_{\text{max}} (\Psi'; N-1).
\end{eqnarray}
On the other hand, if $\Psi'' = \psi_2 \wedge \ldots \wedge \psi_N$ is the optimal Slater approximation of $\Psi'$, then
\begin{eqnarray}\label{ineqn2}
I_{\text{max}} (\Psi'; N-1) = |\langle \Psi'' | \Psi' \rangle |^2 &=& |\langle f\wedge \Psi'' |f\wedge \Psi' \rangle |^2 \nonumber \\
&\leq & I_{\text{max}} (\Psi; N).
\end{eqnarray}
Equation (\ref{factor3}) follows (\ref{ineqn1}) and (\ref{ineqn2}).
Finally, to have the equalities in (\ref{ineqn1}) fulfilled, we must have $\phi_1 = f$ and $\Phi = S'$. That is, $S = f \wedge S'$.
\end{proof}
Lemma \ref{factor2} means that, if some orbital is always occupied, it must be taken so as to maximize the overlap. The fact $S = f\wedge S'$ allows us to factorize the  orbital $f$ out to reduce the $N$-fermion wave function to a $(N-1)$-fermion wave function.

\begin{Lemma}\label{branch2}
Let $\Psi$ be a normalized $N$-fermion wave function and let $f$ and $g$ be two orthogonal normalized orbitals. Suppose $\Psi$ can be decomposed as $\Psi = \Psi' + f \wedge g \wedge \Psi''$,
where $\Psi'$ is a $N$-fermion wave function and $\Psi''$ a $(N-2)$-fermion wave function, and $f$, $g \perp \mathcal{E}(\Psi') \oplus \mathcal{E}(\Psi'')$, then
\begin{eqnarray}\label{branch3}
I_{\text{max}} (\Psi; N) = \max \{ I_{\text{max}} (\Psi'; N),      I_{\text{max}} (\Psi''; N-2)     \}.\; \;\;
\end{eqnarray}
Moreover, if $ I_{\text{max}} (\Psi'; N) >    I_{\text{max}} (\Psi''; N-2)   $, then the optimal Slater approximation $S= S'$, where $S'$ is the optimal Slater approximation of $\Psi'$; while if  $ I_{\text{max}} (\Psi'; N) <   I_{\text{max}} (\Psi''; N-2)   $, then $S= f\wedge g \wedge S''$, where $S''$ is the optimal Slater approximation of $\Psi''$.
\end{Lemma}
\begin{proof}
Let $S$ be the optimal Slater approximation of $\Psi$. Similar to (\ref{ortho1}), we can find an orthonormal basis $\{\phi_1, \phi_2, \ldots, \phi_N\}$ for its supporting space $V$ such that
\begin{subequations}\label{ortho2}
\begin{eqnarray}
\langle f | \phi_i \rangle=0, & & \quad 2\leq i \leq N, \\
\langle g | \phi_i \rangle = 0, & & \quad 3 \leq i \leq N .
\end{eqnarray}
\end{subequations}
This can be done by first determining $\phi_1$ by the projection of $f$ onto $V$, and then determining $\phi_2$ to be the projection of $g$ onto the orthogonal supplement space of $\phi_1$ with respect to $V$.

 Expand $\phi_1$ and $\phi_2$ as $\phi_1 = s_1 f + s_2 g + \tilde{\phi}_1$, $ \phi_2 = t_1 g + \tilde{\phi}_2$,
with $s_1 = \langle f | \phi_1 \rangle $, $s_2 = \langle g | \phi_1 \rangle $, and $t_1 = \langle g | \phi_2 \rangle $. Apparently,
\begin{eqnarray}\label{norm}
|s_1|^2 + |s_2|^2 + \parallel \tilde{\phi}_1 \parallel^2 =1, \quad |t_1|^2+\parallel \tilde{\phi}_2 \parallel^2 =1 .
\end{eqnarray}
We have then $S = (s_1 f + s_2 g + \tilde{\phi}_1)\wedge (t_1 g + \tilde{\phi}_2)\wedge \Phi$,
with $\Phi = \phi_3 \wedge \ldots \wedge \phi_N $. We have
\begin{eqnarray}\label{fg}
I_{\text{max}} (\Psi; N) &=&  |\langle S | \Psi \rangle |^2 \nonumber \\
&=&  |\langle \tilde{\phi}_1 \wedge \tilde{\phi}_2 \wedge \Phi |\Psi'\rangle + s_1^* t_1^* \langle \Phi | \Psi'' \rangle |^2 \nonumber \\
&\leq & \mathcal{M} \big|\parallel \tilde{\phi}_1 \parallel\parallel \tilde{\phi}_2 \parallel+ |s_1||t_1| \big|^2 \nonumber \\
& \leq & \mathcal{M} (\parallel \tilde{\phi}_1 \parallel^2 + |s_1|^2)(\parallel \tilde{\phi}_2 \parallel^2 + |t_1|^2) \nonumber \\
& \leq & \mathcal{M}.
\end{eqnarray}
Here $\mathcal{M}\equiv  \max \{ I_{\text{max}} (\Psi'; N),      I_{\text{max}} (\Psi''; N-2)     \} $. In the fourth line, we used the Cauchy-Schwarz inequality, while in the fifth line, we used (\ref{norm}).
On the other hand, apparently, both $I_{\text{max}} (\Psi'; N)$ and $I_{\text{max}} (\Psi''; N-2)$ are attainable. Therefore, we have (\ref{branch3}).

If $ I_{\text{max}} (\Psi'; N) >    I_{\text{max}} (\Psi''; N-2)   $, it is readily seen that to make the equality in the third line of (\ref{fg}) satisfied, we must have $\parallel \tilde{\phi}_1 \parallel^2= \parallel \tilde{\phi}_2 \parallel^2 =1$, and accordingly $s_1 =s_2= t_1 =0$. Thus, $|\langle S | \Psi \rangle |^2 = |\langle S | \Psi'\rangle |^2$ and hence $S= S'$, where $S'$ is the optimal Slater approximation of $\Psi'$. On the other hand, if  $ I_{\text{max}} (\Psi'; N) <    I_{\text{max}} (\Psi''; N-2)   $,
we must have $|s_1|= |t_1|=1$ and accordingly $\parallel \tilde{\phi}_1 \parallel^2= \parallel \tilde{\phi}_2 \parallel^2 =s_2 =0 $, and thus $S= f\wedge g \wedge S''$, where $S''$ is the optimal Slater approximation of $\Psi''$.
\end{proof}
Lemma \ref{branch2} means that, if in a wave function, two orbitals are occupied or unoccupied only simultaneously, then the wave function can be broken down into two parts---in one the two orbitals are both occupied, while in the other the two are both unoccupied. The two parts can then be studied separately as they do not interfere with each other.

\subsection{Simple applications}\label{examples}

The lemmas above, though simple, are very useful. In the following, we show some examples, for which the lemmas allow us to find $I_\text{max}$ without effort.

\textit{Example 1}. Consider the Greenberger-Horne-Zeilinger (GHZ) -type state ($N\geq 2$)
\begin{eqnarray}\label{cat}
\Psi = a |12 \ldots N \rangle + b | N+1, N+2, \ldots, 2 N \rangle,
\end{eqnarray}
with $|a|^2 + |b|^2 =1$ and $|a|> |b|$. Here and henceforth, by $|ij,\ldots k \rangle $ we mean $\phi_i \wedge \phi_j  \wedge \ldots \wedge \phi_k $. By Lemma \ref{branch2}, we know immediately that the optimal Slater approximation is $|12\ldots N \rangle $, and $I_\text{max}(\Psi;N)= |a|^2$. By Lemma \ref{least}, we know for $M = 2N-1$, we can drop the $2N$th orbital, and $I_\text{max}(\Psi;2N-1)= |a|^2$. Apparently, $I_\text{max}(\Psi;M)$ is a monotonically increasing function of $M$. Therefore, $I_{\text{max}}(\Psi;M )= |a|^2$ for $N\leq M \leq 2N-1$.

\textit{Example 2}. Suppose ($N\geq 3$)
\begin{eqnarray}
\Psi = a |12 \ldots N \rangle + b | 1, N+1, \ldots, 2 N-1 \rangle,
\end{eqnarray}
with $|a|^2 + |b|^2 =1$ and $|a|>|b|$. Now the first orbital is always occupied and by Lemma \ref{factor2}, the problem can be reduced to $\Psi' = a |2 \ldots N \rangle + b | N+1, \ldots, 2 N-1 \rangle $, which is in the form of the wave function in Example 1. We have thus $I_\text{max}(\Psi;N)= |a|^2$ and the optimal Slater approximation is $|12 \ldots N \rangle  $.

\textit{Example 3}. Consider the three-fermion wave function
\begin{eqnarray}
\Psi = a |123\rangle + b|345\rangle + c|567 \rangle .
\end{eqnarray}
Taking $f = |1\rangle $ and $g = |2 \rangle$, the wave function is first decomposed as the sum of $a|123\rangle $ and $b|345\rangle + c|567 \rangle $. For the latter, taking
$f = |3\rangle $ and $g = |4 \rangle$, it is further decomposed as the sum of $b|345\rangle $ and $c | 567 \rangle $. By repeated application of Lemma \ref{branch2}, we readily see that $I_\text{max}(\Psi;3)= \max \{ |a|^2,|b|^2,|c|^2 \}$.

Similarly, for the four-fermion wave function
\begin{eqnarray}
\Psi = a |1234\rangle + b|4567\rangle + c|7891 \rangle ,
\end{eqnarray}
$I_\text{max}(\Psi;4)= \max \{ |a|^2,|b|^2,|c|^2 \}$.

\textit{Example 4}. Consider the two-fermion wave function
\begin{eqnarray}
\Psi = a |12 \rangle + b | 23 \rangle + c | 31 \rangle ,
\end{eqnarray}
with $|a|^2 + |b|^2 + |c|^2 =1$.
Since it is a wave function with two fermions in three orbitals, by Lemma \ref{hole}, $I_\text{max}(\Psi; 2)= 1$.

\section{The simplest nontrivial case: three-in-six}\label{threeinsix}

So far, we have been dealing with the general case. Now we turn to the specific case of three fermions in six orbitals, i.e., $(N,d)= (3,6)$. Like in the ``$N$-representability'' problem, this case is the simplest nontrivial case. The reason is as follows. The $N=2$ case has been solved completely already \cite{zjm}---just take the $M$ most occupied natural orbitals. For $N=3$, the $d=4$ is trivial by Lemma \ref{hole}---the wave function is always a Slater determinant. The $d=5$ case is not difficult as well, since by the particle-hole transform, it can be reduced to the $(N,d)=(2,5)$ case. Or mathematically, there is a Hodge dual between $\wedge^3 \mathcal{H}_5$ and $\wedge^2 \mathcal{H}_5$.

We have the following proposition:
\begin{Proposition}\label{ansatz}
For an arbitrary normalized wave function $\Psi \in \wedge^3 \mathcal{H}_6$, its optimal Slater approximation $S$ is of the form
\begin{eqnarray}\label{alexansatz}
(\alpha_1 \phi_1 + \beta_1 \phi_6)\wedge  (\alpha_2 \phi_2 + \beta_2 \phi_5) \wedge  (\alpha_3 \phi_3 + \beta_3 \phi_4),\quad
\end{eqnarray}
where $\phi_i$ is the natural orbital corresponding to the $i$th largest eigenvalue
of the one-particle reduced density matrix of $\Psi$. The parameters $\alpha_i$ and $\beta_j$ satisfy the condition $|\alpha_i |^2 + |\beta_i |^2=1$.
\end{Proposition}

This proposition is based on the canonical form of the wave function in (\ref{canonical}) below, which was first proven by Chen \textit{et al.} in \cite{lin2015,lin3in6}. Their proof involves sophisticated invariant theory. Here we provide a direct and elementary proof.

\begin{proof}
Let $S = f_1 \wedge f_2 \wedge f_3$ be the optimal Slater approximation of $\Psi$, with $\{ f_1, f_2, f_3\}$ being three orthonormal orbitals. These orbitals can be extended into a full orthonormal basis $\{ f_1, f_2, f_3, g_1, g_2, g_3 \}$ of $\mathcal{H}_6$.

Now let us expand $\Psi$ in terms of the 20 Slater determinants constructed out of $f_i$ and $g_j$. It should be in the form of
\begin{eqnarray}\label{expan1}
\Psi &=& a f_1 \wedge f_2 \wedge f_3 + c g_1 \wedge g_2 \wedge g_3 \nonumber \\
&& + f_1 \wedge (b_{11} g_2\wedge g_3+ b_{12} g_3\wedge g_1 + b_{13} g_1\wedge g_2) \nonumber \\
& & + f_2 \wedge (b_{21} g_2\wedge g_3+ b_{22} g_3\wedge g_1 + b_{23} g_1\wedge g_2) \nonumber \\
& & + f_3 \wedge (b_{31} g_2\wedge g_3+ b_{32} g_3\wedge g_1 + b_{33} g_1\wedge g_2).\quad
\end{eqnarray}
The point is that terms $f_i \wedge f_j \wedge g_k $ do not appear.
The reason is that if some $f_i \wedge f_j \wedge g_k $ has a nonzero coefficient, it would mean a contradiction with the assumption that $f_1 \wedge f_2 \wedge f_3$ is the optimal Slater approximation of $\Psi$. For example, if $f_1 \wedge f_2 \wedge g_1$ has a nonzero coefficient $d$, then $a f_1 \wedge f_2 \wedge f_3$ and $df_1 \wedge f_2 \wedge g_1 $ can be combined into  $f_1 \wedge f_2 \wedge (a f_3 + d g_1 )$, and the Slater wave function $f_1 \wedge f_2 \wedge \tilde{f}_3$,
with $\tilde{f}_3 = (af_3 + d g_1 )/\sqrt{|a|^2+ |d|^2}$, has a larger overlap with $\Psi$ than $f_1 \wedge f_2 \wedge f_3 $.

Next let us simplify the last three terms in (\ref{expan1}). Define
\begin{eqnarray}
(G_1, G_2, G_3) = (g_2 \wedge g_3, g_3 \wedge g_1, g_1 \wedge g_2).
\end{eqnarray}
They are orthonormal. The sum of the last three terms in (\ref{expan1}) is then in the form
\begin{eqnarray}
\Psi' = \sum_{i,j=1}^3 b_{ij} f_i \wedge G_j.
\end{eqnarray}
Now make the transforms
\begin{eqnarray}
f_i = \sum_{j=1}^3 U_{ik }^* e_k , \quad G_j = \sum_{m=1}^3 V_{jm} H_m,
\end{eqnarray}
where $U$ and $V $ are unitary matrices fulfilling the singular value decomposition of the matrix $b$, i.e., $U^\dagger b V = \Lambda $, where $\Lambda$ is diagonal.
Under these transforms,
\begin{eqnarray}\label{psiprime}
\Psi' =  \sum_{k=1}^3 \Lambda_{kk } e_k \wedge H_k .
\end{eqnarray}
Here $H_k$, as a linear combination of $G_j$, by Lemma \ref{hole}, can be written in the form
\begin{eqnarray}
(H_1, H_2, H_3 )= (H_{12}\wedge H_{13},H_{23}\wedge H_{21},H_{31}\wedge H_{32}),\quad
\end{eqnarray}
with $H_{ij}\in  W \equiv \Span\{ g_1, g_2, g_3 \}$ and normalized. Moreover,
\begin{eqnarray}
\langle H_{12}| H_{13}\rangle  =\langle H_{21}| H_{23}\rangle =\langle H_{31}| H_{32}\rangle =0.
\end{eqnarray}
We can expand $\{ H_{12},H_{13}\}$, $\{ H_{21},H_{23}\}$, and $\{ H_{31},H_{32}\}$ into a complete orthonormal basis of $W $,
\begin{eqnarray}
W & =& \Span \{ H_{12},H_{13}, h_1 \} \nonumber \\
&=& \Span \{ H_{21},H_{23}, h_2 \} \nonumber \\
&=& \Span \{ H_{31},H_{32}, h_3 \} .
\end{eqnarray}
From the fact that
\begin{eqnarray}
\langle H_1 | H_2 \rangle = \langle H_2 | H_3\rangle =\langle H_1 | H_3 \rangle =0,
\end{eqnarray}
it is easy to deduce that
\begin{eqnarray}
\langle h_1 | h_2 \rangle = \langle h_2 | h_3\rangle =\langle h_1 | h_3 \rangle =0.
\end{eqnarray}
That is, $\{ h_1 , h_2, h_3\}$ is an orthonormal basis of $W$. Thus,
\begin{eqnarray}
H_1 \propto h_2\wedge h_3 , \; H_2 \propto h_3 \wedge h_1, \; H_3 \propto h_1 \wedge h_2, \quad \;  \label{h1h2h31}\\
g_1 \wedge g_2 \wedge g_3 \propto h_1 \wedge h_2 \wedge h_3.\quad \quad \quad \quad \quad  \label{h1h2h32}
\end{eqnarray}
Similarly, since $\Span \{ e_1, e_2, e_3 \} = \Span \{ f_1, f_2, f_3 \}$,
\begin{eqnarray}\label{f1f2f3}
f_1 \wedge f_2 \wedge f_3 \propto e_1 \wedge e_2 \wedge e_3.
\end{eqnarray}
Finally, substituting (\ref{psiprime}), (\ref{h1h2h31})-(\ref{f1f2f3}) into (\ref{expan1}), we have the canonical form of $\Psi$,
\begin{eqnarray}\label{canonical}
\Psi &=& A e_1 \wedge e_2 \wedge e_3 + B e_1 \wedge h_2 \wedge h_3 + C e_2 \wedge h_3 \wedge h_1  \nonumber \\
& & + D e_3 \wedge h_1 \wedge h_2 + E h_1 \wedge h_2 \wedge h_3,
\end{eqnarray}
with $\{ e_1, e_2, e_3, h_1, h_2, h_3 \}$ being an orthonormal basis of $\mathcal{H}_6$, $ e_1 \wedge e_2 \wedge e_3 $ the optimal Slater approximation of $\Psi$ and $|A|^2+ |B|^2+ |C|^2+ |D|^2+|E|^2=1$.

Now, it is readily seen that, the one-particle reduced density matrix $\rho$ in the basis of $\{ e_1, h_1, e_2,h_2, e_3, h_3 \}$ is a block matrix consisting of three $2\times 2$ matrices. Specifically, the $2\times 2$ block corresponding to $\{ e_1, h_1 \}$ is
\begin{eqnarray}
\left( \begin{array}{cc }
|A|^2 + |B|^2  &  E^* B \\
E B^* &  |C|^2+ |D|^2+|E|^2
\end{array}    \right),
\end{eqnarray}
the $2\times 2$ block corresponding to $\{ e_2, h_2 \}$ is
\begin{eqnarray}
\left( \begin{array}{cc }
|A|^2 + |C|^2  &  E^* C \\
E C^* &  |B|^2+ |D|^2+|E|^2
\end{array}    \right),
\end{eqnarray}
and the $2\times 2$ block corresponding to $ \{ e_3, h_3 \}$ is
\begin{eqnarray}
\left( \begin{array}{cc }
|A|^2 + |D|^2  &  E^* D \\
E D^* &  |B|^2+ |C|^2+|E|^2
\end{array}    \right).
\end{eqnarray}
An important fact is that the three matrices all have trace of unity, which means that the sum of the two eigenvalues is unity. Therefore, we have
\begin{eqnarray}
\lambda_1 + \lambda_6  = \lambda_2 + \lambda_5 =\lambda_3 + \lambda_4 =1.
\end{eqnarray}
We have thus proven that Borland-Dennis's discovery (\ref{borlandanddennis}a) is necessary.

In the generic case, there is no degeneracy between the eigenvalues and
 therefore $\{ \phi_1 ,\phi_6 \} $ must appear in the same block, and so must $\{ \phi_2, \phi_5 \}$ and $\{ \phi_3, \phi_4 \}$.
This means, for some permutation $P $ in the group $ S_3$,
\begin{subequations}
\begin{eqnarray}
e_{P1} &=& \alpha_1 \phi_1 + \beta_1 \phi_6, \\
e_{P2} &=& \alpha_2 \phi_2 + \beta_2 \phi_5, \\
e_{P3} &=& \alpha_3 \phi_3 + \beta_3 \phi_4,
\end{eqnarray}
\end{subequations}
with $|\alpha_i |^2 + |\beta_i |^2=1$. We thus have (\ref{alexansatz}).

In the case of degeneracy, the optimal Slater approximation can still be written in the form of (\ref{alexansatz}), but the natural orbitals can no longer be chosen in an arbitrary way. For example, for the GHZ state
\begin{eqnarray}
\Psi_{\text{GHZ}} =\frac{1}{\sqrt{2}}(\phi_1 \wedge \phi_2 \wedge \phi_3 + \phi_4 \wedge \phi_5 \wedge \phi_6),
\end{eqnarray}
its optimal Slater determinant is either $\phi_1 \wedge \phi_2 \wedge \phi_3$ or $\phi_4 \wedge \phi_5 \wedge \phi_6 $ by Lemma \ref{branch2}. However, since the one-particle reduced density matrix is simply a constant matrix with $1/2$ on the diagonal, an arbitrary vector is an natural orbital. It is obvious that we cannot choose the natural orbitals randomly.
\end{proof}

Here some remarks are worth mentioning. Equation (\ref{alexansatz}), initially as an ansatz, was motivated by two observations of Borland and Dennis. The first one is Eq.~(\ref{borlandanddennis}a). In the expansion of (\ref{alexansatz}), in each term, one and only one of $\{ \phi_1, \phi_6 \}$ appear, which means their occupations sum up to unity. This is in alignment with (\ref{borlandanddennis}a). The second one is that, in the expansion of (\ref{alexansatz}), we have eight terms, which are exactly the nonzero terms appearing in the expansion of a generic function $\Psi$
in terms of its natural orbitals \cite{borland1, borland2}. These two facts make the ansatz (\ref{alexansatz}) plausible and promising.

Proposition \ref{ansatz} also clarifies the role of the natural orbitals in the construction of the optimal Slater approximation of an arbitrary wave function. As pointed out in Ref.~\cite{zjm}, generally the Slater wave function constructed out of the $N$ most occupied natural orbitals shows no definite advantage over a Slater wave function built of $N$ randomly generated orbitals. Actually, in some cases---for example, the state in Eq.~(\ref{patrik}) below---the former can be orthogonal to the target wave function. The usefulness of natural orbitals for generating an efficient expansion of the target wave function was also seriously questioned in Ref.~\cite{holland}. Now, in the special case of $\wedge^3 \mathcal{H}_6$,   Proposition \ref{ansatz} tells us that the natural orbitals do play an important role in the optimal Slater approximation, but in a subtle way.

Although Eq.~(\ref{alexansatz}) is beautiful and general, it is not convenient as a variational wave function, as it involves the natural orbitals. To apply it for finding the optimal Slater approximation, one must first solve the natural orbitals, which is not necessarily an easy task, and then expand the wave function in terms of the natural orbitals, which is again tedious if it is feasible at all.

However, for some special kind of states, Eq.~(\ref{alexansatz}) can be expressed in another form without explicitly referring to the natural orbitals. Suppose we have a wave function in the form of
\begin{eqnarray}\label{eight}
\Psi &=& A_{0} \varphi_1 \wedge \varphi_2 \wedge \varphi_3 + A_{1} \varphi_1 \wedge \varphi_2 \wedge \varphi_4 + A_{2} \varphi_1 \wedge \varphi_5 \wedge \varphi_3  \nonumber \\
& + &  A_{3} \varphi_1 \wedge \varphi_5 \wedge \varphi_4 + A_{4} \varphi_6 \wedge \varphi_2 \wedge \varphi_3 + A_{5} \varphi_6 \wedge \varphi_2 \wedge \varphi_4  \nonumber \\
& +& A_{6} \varphi_6 \wedge \varphi_5 \wedge \varphi_3 + A_{7} \varphi_6 \wedge \varphi_5 \wedge \varphi_4.
\end{eqnarray}
Here $\{ \varphi_i | 1\leq i \leq 6\}$ is an orthonormal basis of $\mathcal{H}_6$ and $\sum_{j=0}^7 |A_j|^2=1$. We note that the six basis vectors are divided into three pairs, i.e., $\{\varphi_1, \varphi_6 \}$, $\{\varphi_2, \varphi_5 \}$, and $\{\varphi_3, \varphi_4 \}$. The eight Slater determinants are constructed by choosing one out of each pair. We also note that the canonical form (\ref{canonical}) is a special case of (\ref{eight}). Now, like above, it is readily seen that the one-particle density matrix $\rho$ has a block diagonal form with respect to the basis  $\{ \varphi_i | 1\leq i \leq 6\}$. More specifically, there are three $2\times 2$ blocks corresponding to the three pairs, and each block has trace unity. In a generic case without degeneracy between the eigenvalues, we then must have
\begin{subequations}
\begin{eqnarray}
\Span \{ \varphi_1, \varphi_6 \} &=& \Span \{\phi_{P1}, \phi_{7-P1} \}, \\
\Span \{ \varphi_2, \varphi_5 \} &=& \Span \{\phi_{P2}, \phi_{7-P2} \}, \\
\Span \{ \varphi_3, \varphi_4 \} &=& \Span \{\phi_{P3}, \phi_{7-P3} \},
\end{eqnarray}
\end{subequations}
for some permutation $P \in S_3$. This means the optimal Slater approximation can also be expressed in terms of the orbitals $\varphi_i$ as
\begin{eqnarray}\label{alexansatz2}
(\alpha_1 \varphi_1 + \beta_1 \varphi_6)\wedge  (\alpha_2 \varphi_2 + \beta_2 \varphi_5) \wedge  (\alpha_3 \varphi_3 + \beta_3 \varphi_4).\quad
\end{eqnarray}
The problem is then to maximize the magnitude of
\begin{eqnarray}
\langle \Psi | S \rangle &=& A_0^* \alpha_1 \alpha_2 \alpha_3 + A_1^* \alpha_1 \alpha_2 \beta_3 - A_2^* \alpha_1 \beta_2 \alpha_3 \nonumber \\
& & - A_3^* \alpha_1 \beta_2 \beta_3 + A_4^* \beta_1 \alpha_2 \alpha_3 + A_5^* \beta_1 \alpha_2 \beta_3 \nonumber \\
& & - A_6^* \beta_1 \beta_2 \alpha_3 - A_7^* \beta_1 \beta_2 \beta_3,
\end{eqnarray}
under the condition $|\alpha_i|^2+|\beta_i|^2 =1 $.

\section{Application of Hadamard inequality }\label{hadamardsection}

Neither (\ref{alexansatz}) nor (\ref{alexansatz2}) are useful in case of degeneracy,
because it is unclear what specific natural orbitals one should take. For example, for the state
\begin{eqnarray}\label{patrik}
\Psi &=& \frac{1}{\sqrt{3}} (\varphi_1 \wedge \varphi_2 \wedge \varphi_4 + \varphi_1 \wedge \varphi_5 \wedge \varphi_3   \nonumber \\
& & \quad \quad
 +  \varphi_6 \wedge \varphi_2 \wedge \varphi_3  ),\quad
\end{eqnarray}
$\{\varphi_1,\varphi_2, \varphi_3 \}$ are degenerate with occupation of $2/3$, and $\{\varphi_4,\varphi_5, \varphi_6 \}$ are degenerate with occupation of $1/3$. If one takes the ansatz as
\begin{eqnarray}\label{alexansatz3}
(\alpha_1 \varphi_1 + \beta_1 \varphi_4)\wedge  (\alpha_2 \varphi_2 + \beta_2 \varphi_6) \wedge  (\alpha_3 \varphi_3 + \beta_3 \varphi_5),\quad
\end{eqnarray}
one always gets $\langle S | \Psi \rangle = 0 $.

There is a method based on the Hadamard inequality \cite{horn} to determine $I_\text{max}$ and construct the optimal Slater approximation of $\Psi$. Let $S= \phi_1 \wedge \phi_2 \wedge \phi_3$ be the optimal Slater approximation, with the orthonormal orbitals $\phi_i $ related to  $\varphi_j$ by a $3\times 6 $ matrix $M$, i.e.,
\begin{eqnarray}
\phi_i = \sum_{j=1}^6 M_{ij} \varphi_j.
\end{eqnarray}
Apparently, $M$ should satisfy the condition $MM^\dagger = I_{3\times 3 }$, with $ I_{3\times 3 } $ being the $3\times 3$ identity matrix. We have
\begin{eqnarray}
\langle \Psi | S \rangle = \frac{1}{\sqrt{3}}\left( \det M_{[124]} + \det M_{[623]}+ \det M_{[153]} \right),
\end{eqnarray}
where $M_{[ijk]}$ denotes the $3\times 3$ matrix composed of column $i$, $j$, $k$ of $M$.
Now by the Hadamard inequality for the determinant of a matrix, we have
\begin{eqnarray}\label{hada1}
|\langle \Psi | S \rangle | &\leq & \frac{1}{\sqrt{3}}\left( \left| \det M_{[124]} \right | + \left| \det M_{[623]} \right |+ \left | \det M_{[153]} \right | \right)\nonumber \\
&\leq &  \frac{1}{\sqrt{3}}\left( a_1 a_2 a_4 + a_2 a_3 a_6 + a_1 a_3 a_5  \right),
\end{eqnarray}
where $a_i $ is the norm of the $i$th column of $M$. They must satisfy the conditions, $0\leq a_i \leq 1$ and $\sum_{i=1}^6 a_i^2 = 3$. Now by using the Cauchy-Schwarz inequality and the arithmetic mean-geometric mean inequality, we have
\begin{eqnarray}
& & a_1 a_2 a_4 + a_2 a_3 a_6 + a_1 a_3 a_5   \nonumber \\
&\leq & \sqrt{a_1^2 a_2^2 + a_2^2 a_3^2 + a_1^2 a_3^2} \sqrt{a_4^2 + a_5^2 + a_6^2} \nonumber \\
&\leq & \sqrt{\frac{1}{3}( a_1^2 + a_2^2 + a_3^2)^2(3- a_1^2 - a_2^2 - a_3^2)} \nonumber \\
&= & \sqrt{\frac{1}{6} ( a_1^2 + a_2^2 + a_3^2)^2(6- 2a_1^2 -2 a_2^2 - 2a_3^2)}\nonumber \\
& \leq&  \frac{2}{\sqrt{3}}. \quad \;
\end{eqnarray}
The equalities are achieved when and only when $a_1= a_2 = a_3 = \sqrt{2/3}$ and $a_4 = a_5 = a_6 = \sqrt{1/3}$. By (\ref{hada1}), we thus get $|\langle \Psi | S \rangle | \leq 2/3$. Now fortunately, the equalities in (\ref{hada1}) can all be satisfied if one take
\begin{eqnarray}
M = \left(\begin{array}{cccccc}
\sqrt{\frac{2}{3}} & 0 & 0 & 0 & 0 & \sqrt{\frac{1}{3}}e^{i\theta} \\
0 & \sqrt{\frac{2}{3}} & 0 & 0   & \sqrt{\frac{1}{3}}e^{i\theta} & 0  \\
0 & 0 & \sqrt{\frac{2}{3}} &  \sqrt{\frac{1}{3}}e^{i\theta} & 0 & 0
\end{array} \right),\quad
\end{eqnarray}
with $\theta\in \mathbb{R}$ being arbitrary.
Therefore, we have
\begin{eqnarray}
I_{\text{max}} (\Psi; 3 )= \frac{4}{9},
\end{eqnarray}
and the optimal Slater approximation is
\begin{eqnarray}\label{alexansatz4}
(\alpha \varphi_1 + \beta \varphi_6)\wedge  (\alpha \varphi_2 + \beta \varphi_5) \wedge  (\alpha \varphi_3 + \beta \varphi_4),\quad
\end{eqnarray}
with $\alpha = \sqrt{2/3}$ and $\beta = \sqrt{1/3} e^{i\theta }$. We note that as proven by Chen \textit{et al.} \cite{lin3in6}, $4/9$ is the minimum value achievable by $I_\text{max}$ for three fermions in six orbitals.

Similarly, by using the Hadamard inequality, we can show that for
\begin{eqnarray}
\Psi &=& \frac{1}{2} \big[ \varphi_1 \wedge \varphi_2 \wedge \varphi_3 + \varphi_1 \wedge \varphi_5 \wedge \varphi_4  \nonumber \\
& & \quad \quad + \varphi_6 \wedge \varphi_5 \wedge \varphi_3 + \varphi_6 \wedge \varphi_2 \wedge \varphi_4 \big ],
\end{eqnarray}
$I_\text{max}(\Psi;3)= 1/2$, and the optimal Slater approximation is in the form of (\ref{alexansatz4}), with $\alpha =\beta =1/ \sqrt{2} $.
For the cyclic state
\begin{eqnarray}
\Psi &=& \frac{1}{\sqrt{6}} [ \varphi_1 \wedge \varphi_2 \wedge \varphi_3 + \varphi_2 \wedge \varphi_3 \wedge \varphi_4 + \varphi_3 \wedge \varphi_4 \wedge \varphi_5 \nonumber \\
& &+ \varphi_4 \wedge \varphi_5 \wedge \varphi_6 + \varphi_5 \wedge \varphi_6 \wedge \varphi_1 + \varphi_6 \wedge \varphi_1 \wedge \varphi_2 ],\quad \;
\end{eqnarray}
$I_\text{max}(\Psi;3)= 3/4$, and the optimal Slater approximation is
\begin{eqnarray}\label{alexansatz4}
(\alpha \varphi_1 + \beta \varphi_4)\wedge  (\alpha \varphi_2 + \beta \varphi_5) \wedge  (\alpha \varphi_3 + \beta \varphi_6),\quad
\end{eqnarray}
with $\alpha = \beta =1/ \sqrt{2} $.

The Hadamard inequality can actually be employed to solve Examples 1, 2, and 3 in Sec.~\ref{examples} too. For example, consider the GHZ-type state in (\ref{cat}) with $|a|>|b|$. Let
\begin{eqnarray}
\phi_i = \sum_{j=1}^{2N} M_{ij} \varphi_j , \quad 1\leq i \leq N,
\end{eqnarray}
be the $N$ orthonormal orbitals supporting the optimal Slater approximation $S$. Let $M_1$ and $M_2$ be the $N\times N$ matrices formed by the first  and the last $N$ columns of $M$, respectively, and let $c_i$ and $d_i$ be the norm of the $i$th row of $M_1$ and $M_2$, respectively. We have $c_i^2 + d_i^2 =1$, and
\begin{eqnarray}
|\langle \Psi | S \rangle | &\leq & |a| c_1 c_2\ldots c_N + |b| d_1 d_2\ldots d_N \nonumber \\
&\leq & |a| (c_1 c_2\ldots c_N  +  d_1 d_2\ldots d_N) \nonumber \\
&\leq & |a| (c_1 c_2 + d_1 d_2 ) \nonumber \\
&\leq & |a| \sqrt{c_1^2 + d_1^2}  \sqrt{c_2^2 + d_2^2} \nonumber \\
& \leq &  |a| .
\end{eqnarray}
The equalities are achieved if and only if $c_1 =c_2 =\ldots  c_N = 1$.

\section{Conclusions and outlooks}\label{conclusion}

We have studied analytically the problem of optimally approximating a fermionic wave function by a Slater determinant.
It is a continuation of the work in Ref.~\cite{zjm}, where a numerical algorithm was proposed and the problem in the $N=2$ case was solved.
This problem is of significance in both  physics and mathematics.
Physically, it provides a geometric measure of entanglement between the fermions, with identical particles treated identically.
Mathematically, it is about the geometry of the alternating space $\wedge^N \mathcal{H}_d$ equipped with an inner product.
We found that the geometric structure of $\wedge^N \mathcal{H}_d$ is subtle and interesting,
as illustrated by the lemmas in the general case, and the proposition on the structure of the optimal Slater approximation
in the special case of $(N=3,d=6)$.

Our simple proof of the Borland-Dennis discovery (\ref{borlandanddennis}a) by starting from the optimal Slater approximation
also exemplifies that this notion can play a pivotal role in analyzing the structure of a fermionic wave function.

We consider our results as tentative in this direction. Many problems are still open.
For example, a question of interest is the minimum of $I_\text{max}$ and its scaling behavior.
The significance of this min-max problem lies in that it characterizes the complexity
of wave functions in $\wedge^N \mathcal{H}_d$ as a whole.
For a general value of $d$ and with $N=3$, it is easy to prove that the minimum of $I_\text{max}$ is no less than $2/d^2$,
in contrast to the dimension of $\wedge^3 \mathcal{H}_d$, which scales as $d^3$.
But it is unclear whether the power factor of 2 can be reduced or not.

\section*{Acknowledgments}

We are grateful to Alex D. Gottlieb for his essential input. 
We acknowledge financial support by the Fujian Provincial 
Science Foundation under Grant No. 2016J05004, by the Austrian Ministry of Science, Research and Economy (BM:WFW) 
via its grant for the Wolfgang Pauli Institute and by the 
Austrian Science Foundation (FWF) under Grant No. F41 
(SFB ``VICOM'') as well as under Grant No. W1245 (DK 
``Nonlinear PDEs'').



\begin{thebibliography}{99}
\bibitem{hartreefock}
The optimal Slater approximation we refer to here is not to be confused with the usual Hartree-Fock approximation, which is the ``optimal'' approximation in that it minimizes the ground state energy. The point is that, the Hartree-Fock wave function is always with respect to a Hamiltonia, while in our problem we have no such thing as a Hamiltonian. We need only the notion of the Hilbert space, or more specifically, the inner product. We just analysize the structure of a given wave function in a Hilbert space, regardless of from where it is derived.
\end{thebibliography}
\end{document}